\documentclass[sigplan]{acmart}
\renewcommand\footnotetextcopyrightpermission[1]{}
\usepackage{stmaryrd}
\usepackage{xcolor}
\usepackage{comment}
\usepackage{tikz}
\usepackage{enumitem}
\usetikzlibrary{positioning}

\newcommand{\dsig}[2]{\ensuremath{\textsc{dsig}_{#2}(#1)}}

\title[Trees and Turtles: Modular Abstractions for State Machine Replication Protocols]{Trees and Turtles: Modular Abstractions for State Machine Replication Protocols}
\author{Natalie Neamtu}
\affiliation{
    \institution{Microsoft Corporation}
    \city{Redmond}
    \state{WA} 
    \country{USA}}
\email{nan55@cornell.edu}
\author{Haobin Ni}
\affiliation{
    \institution{Cornell University}
    \city{Ithaca}
    \state{NY} 
    \country{USA}}
\email{haobin@cs.cornell.edu}
\author{Robbert van Renesse}
\affiliation{
    \institution{Cornell University}
    \city{Ithaca}
    \state{NY} 
    \country{USA}}
\email{rvr@cs.cornell.edu}

\begin{abstract}
We present two abstractions for 
designing modular state machine replication (SMR) protocols: 
\emph{trees} and \emph{turtles}.  A tree captures
the set of possible state machine histories, 
while a turtle represents
a subprotocol that tries to find agreement in this tree.
We showcase the applicability of these abstractions by
constructing crash-tolerant SMR protocols out 
of abstract tree turtles
and providing examples of tree turtle implementations.
Tree turtles can also be extended to be made Byzantine fault-tolerant (BFT).
The modularity of tree turtles
allows a generic approach for adding a leader for liveness.
We expect that these abstractions will simplify reasoning
and formal verification of SMR protocols as well as facilitate innovation in protocol designs.
\end{abstract}

\begin{document}

\maketitle

\section{Introduction}

State machine replication (SMR) is a widely-used paradigm 
in distributed and decentralized services, 
wherein a set of processors provides an abstraction of 
a single state machine with an ever-growing history~\cite{S90}. 
In the face of possible processor failures and unbounded communication delays, 
the challenge lies in ensuring that nodes always agree on 
the history 
while allowing updates to be made in as timely a manner as possible. 

Traditional SMR protocols are usually constructed around
the notion of an unbounded sequence of slots. 
The goal of such a protocol is to fill the slots with values.
A typical protocol consists of an unbounded series of
rounds where the contents of
at most one slot may be decided in a round. 
Some protocols are only able to fill a single slot
(i.e., \cite{Ben83,DLS88}),
thus an unbounded number of
instances of the protocol must be used to implement SMR.
Higher throughput is achieved by running multiple instances
in parallel, either independently~\cite{L98,CL99} or
in a pipelined fashion~\cite{ZAB,HotStuff}, or by putting a batch of 
values in each slot.

We present here two abstractions that break this slot-by-slot paradigm: 
\emph{trees} and \emph{turtles}.
Referring to a sequence of values as a \emph{chain},
the set of such chains forms a tree under the \emph{is-a-prefix-of} relation.
As was proposed in~\cite{GENPAXOS}, we generalize the slot-based scheme 
for constructing SMR protocols to
allow entire chains to be decided at once, 
extending the state machine history down a path through the tree.
We then generalize the notion of a round to
an abstract subprotocol which can be used to decide one chain.
We refer to such protocols as turtles because
they are stacked \emph{in infinitum} to construct an SMR protocol.
Taken together, the result is protocols called \emph{tree turtles}.

Using chains requires our subprotocols to form a consensus out of a set with 
richer algebraic structures than the traditional set of singlar values.
In this case, the algebraic structure we use is the meet-semilattice formed
by the ancestor relation between the nodes of the tree as the partial order and
the lowest common ancestor of a set of nodes as the meet operator.
This structure is utilized in the tree turtle protocols 
and their proofs of correctness that we present 
in this paper.

We expect that our abstractions can lead to various advantages
over traditional SMR approaches.
Proposing chains allows processors to specify preferred orderings of values
in the state machine history.
Reasoning about SMR protocols is made simpler with tree turtles
because the never-terminating execution is factored out;
this can lead to more reusable proofs in both an informal an formal setting.
Tree turtles themselves have simple proofs
when compared to existing SMR protocols.
The modularity of tree turtles also enables the design of heterogeneous
protocols that dynamically adapt to their workloads or operating
conditions.

For liveness, we show how to compose our turtle abstraction 
with a \emph{leader} which attempts to eliminate contention between processors.
Different from traditional approaches,
our protocol does not require a non-faulty leader to make progress
under favorable conditions. 
We also demonstrate that tree turtles can be made Byzantine fault-tolerant (BFT)
by adding an abstraction called \emph{evidence} 
which allows protocols to use additional
techniques for restricting Byzantine behavior (such as digital signatures).

\section{Trees and Turtles}

Put simply, the goal of SMR is to allow 
a set of processors agree on an ever-growing sequence of values
in a fault-tolerant manner.
Rather than focusing on the individual values in a sequence,
we will consider how an agreement can be formed on an entire sequence,
or a \emph{chain}, at once. Doing so utilizes the tree structure present
in sets of chains.

Consider a processor $p$ that believes the state machine history
is represented by a chain $c$. If another processor $p'$ 
believes the state machine is represented by a different chain $c'$, 
then $p$ and $p'$ should be able to agree on a common history.
Beyond the simplest case where $c = c'$, 
consider whether one chain is a \emph{prefix} of the other.
If this is the case, the processor who proposed the shorter chain 
could later ``catch up" by extending its chain to the longer chain,
without needing to modify the earlier state machine history.

If $p$ and $p'$ attempt to establish agreement on
an extension of the state machine history by
proposing chains $c$ and $c'$ to each other, then 
we could consider that $p$ and $p'$ agree on the \emph{longest common prefix}
of $c$ and $c'$. This is the (possibly empty)
longest identical subsequence that can be found 
starting from the beginning of each chain. 
The longest common prefix can readily be generalized to any number of chains; 
Figure \ref{fig:trees} depicts a tree formed by three chains
where the longest common prefixes are ancestor nodes.
\begin{figure}[t]
\centering
\begin{tikzpicture}
\node at (0,0) {$\bot$} [grow = east]
    child { node { $d$ } 
        child { node { $c_3$ } }
        child { node { $u$ }
            child {
                child { node { $c_2$ } }
            }
            child { node { $c_1$ } }
        }
    };
\end{tikzpicture} 
\vspace{-1mm}
\caption{Illustrating the tree structure formed by trees and the is-a-prefix-of relation. Chain $u$ is the longest common prefix of chains $c_1$ and $c_2$. Chain $d$ is the longest common prefix of chains $c_1$, $c_2$, and $c_3$.
}
\vspace{-1mm}
\label{fig:trees}
\end{figure}
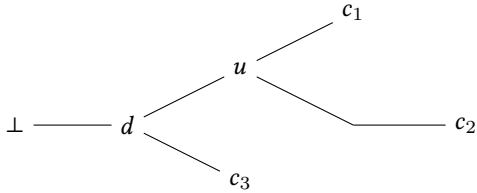

Trees\footnote{
    In addition to being like trees in a loose conceptual sense, 
    chains, taken with the partial order $\preceq$,
    satisfy the set-theoretic definition of a tree.
    This is because for each chain $c$, the set of its prefixes 
    $\{ c' ~|~ c' \preceq c \}$ is totally ordered by $\preceq$.
    As mentioned in the introduction, it is possible to generalize 
    trees even further to a semilattice.}, then, can be viewed as fundamental to state machine replication:
processors continually propose new chains to each other 
in order to keep extending the longest agreed-upon path through the tree.
Table \ref{tab:notation} summarizes the notation which we will
use for chains in this paper.

\begin{table}[h]
\caption {Summary of notation for chains.} \label{tab:notation} 
\begin{center}
\begin{tabular}{c|l}
\hline
$\bot$ & the empty chain \\ \hline
$c \preceq c'$ & $c$ is a \emph{prefix} of $c'$ (or, $c'$ is an \emph{extension} of $c$) \\ \hline
$c \simeq c'$ & $c \preceq c' \text{ or } c' \preceq c$ (we say that $c$ and $c'$ \emph{agree}) \\ \hline
$c \sqcap c'$ & the longest common prefix of $c$ and $c'$ \\ \hline
\end{tabular}
\end{center}
\end{table}

Now we will turn our attention to the structure of an SMR protocol. 
To perform SMR, 
processors alternate between proposing new extensions to the longest
agreed-upon chain, collecting proposals from other processors,
and deciding what the new longest agreed-upon chain is. 
This process repeats forever to ensure that the longest
agreed-upon chain is ever-growing. 

Thus, we may extract a natural building block from this structure: 
a subprotocol in which processors propose 
and subsequently decide on a single chain. 
While many subprotocol implementations
are possible, we establish a common specification
for the properties they should have.
Once we have one (or many) such subprotocols, 
constructing a SMR protocol
is simple: we stack these subprotocols on top of each other in an unending sequence.
Owing to this infinite repetition, and inspired by the saying
\emph{``turtles all the way down,''} we refer to our building blocks of
SMR protocols as \emph{turtles}.
Since our protocols combine the ideas of trees and turtles, 
we will call them \emph{tree turtles} (Figure~\ref{fig:turtles}).

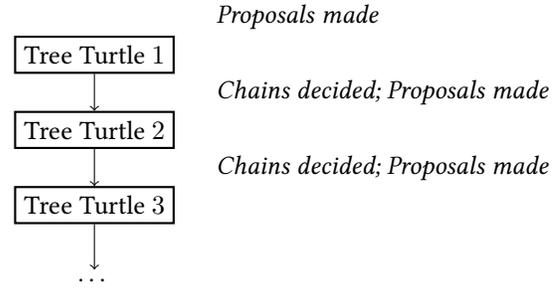
\begin{figure}[t]
\centering
\begin{tikzpicture}[
roundnode/.style={rectangle, draw=black, fill=white, thick},
ionode/.style={rectangle, fill=white},
factnode/.style={rectangle}
]
\node[roundnode]    (turtle1) at (0, 0)   {Tree Turtle $1$};
\node[roundnode]    (turtle2) at (0, -1)   {Tree Turtle $2$};
\node[roundnode]    (turtle3) at (0, -2)   {Tree Turtle $3$};
\node (dummy) at (0, -3) {$\dots$};

\node[factnode, anchor=west] (text1) at (1.5, 0.5) {\textit{Proposals made}};
\node[factnode, anchor=west] (text2) at (1.5, -0.5) {\textit{Chains decided; Proposals made}};
\node[factnode, anchor=west] (text2) at (1.5, -1.5) {\textit{Chains decided; Proposals made}};

\draw[->] (turtle1.south) -- (turtle2.north);
\draw[->] (turtle2.south) -- (turtle3.north);
\draw[->] (turtle3.south) -- (dummy.north);
\end{tikzpicture} 
\vspace{-2mm}
\caption{Tree turtle protocols are stacked to implement a state machine replication protocol.}
\label{fig:turtles}
\vspace{-2mm}
\end{figure}

Tree turtles reduce the problem of solving SMR into making a
single proposal and decision. The specification for tree turtles 
requires them to always terminate (Section~\ref{sec:tt_spec}), while  
SMR requires that the protocol goes on forever (Section~\ref{sec:smr_spec}).
The tree turtle implementations are encapsulated from the correctness of
the SMR protocol proved in Section~\ref{sec:tt_to_smr} making those proofs
reusable across implementations. This also readily demonstrates the
potential of heterogeneous SMR protocols composed from multiple types
of tree turtles.

\section{System Model} \label{sec:sys_model}

We will denote the set of processors participating in a protocol as $\mathcal{P}$. 
Processors can communicate with all other processors by sending messages over a network.

\subsubsection*{Failures.} 
Each processor $p$ can either be correct or faulty. 
Faulty processors may fail by crashing, at which point they may stop 
executing indefinitely.
Which processors are faulty is not known \emph{a priori}. 

\subsubsection*{Network.} 
The network is assumed to be \emph{reliable} in the following sense: 
if a processor $p$ sends a message to a correct processor $q$, 
then $q$ eventually receives that message.
We also assume that the network doesn't forge or garble messages,
meaning that if a processor $q$ receives a message 
from a processor $p$, then $p$ actually sent that message.

\subsubsection*{Asynchrony.} 
We assume that there are no upper bounds on the difference between 
processor speeds or network latency. 
This means that processors are not able to distinguish between 
faulty processors and correct processors that are simply slow 
or whose messages have yet to be delivered.

\subsubsection*{Quorums.} 
A processor can never expect to receive a message
from all other processors due to the possibility that some have crashed. 
To address this, a processor instead waits for a \emph{quorum} 
which is a subset of the processors in $\mathcal{P}$.
A \emph{quorum system} $\mathcal{Q}$ is a set of quorums
($\mathcal{Q} \subseteq 2^\mathcal{P}$) which satisfies
the following: 

\begin{itemize}
    \item We assume that any quorum system $\mathcal{Q}$ 
    contains a special quorum, which we will denote $Q^*$,
    that consists entirely of correct processors. 
    \item We say that a quorum system satisfies $k$-\emph{intersection}
    if any set of $k$ quorums have at least one processor in common~\cite{JM05}.
\end{itemize}

Appendix~A shows how $k-$intersecting quorum systems can be
implemented using threshold quorums if $|\mathcal{P}| > k \cdot f$,
where $f$ is the maximum number of faulty processors.

\section{State Machine Replication Specification} \label{sec:smr_spec}

Below we give the requirements for a SMR protocol,
in which processors in $\mathcal{P}$ alternate \emph{proposing} and \emph{deciding} chains.

\begin{definition}[State Machine Replication Specification] \label{def:smr_spec}
\par\noindent $ $
\begin{itemize}[leftmargin=1.25em]
\item \emph{SMR-Agreement}:
if a processor $p$ decides chain $d$ and a processor $p'$
decides chain $d'$, then $d \simeq d'$ (even if $p = p'$).
\item \emph{SMR-Validity}:
if a processor decides chain $d$, 
then some processor proposed a chain $c$ such that $d \preceq c$.
\item \emph{SMR-Relay}:
if a correct processor decides chain $d$, then eventually all correct processors
decide $d$ or an extension of $d$.
\item \emph{SMR-Monotonicity}:
if a processor decides chain $d$ and later another chain $d'$, then $d \preceq d'$.
\item \emph{SMR-Progress}: 
if a correct processor proposes a chain that it has not decided before, then
that processor eventually decides a chain it has not decided before.
\end{itemize}
\end{definition}
Note that it is impossible for a protocol to satisfy 
all properties if the system is asynchronous~\cite{FLP83}.
For our tree turtle protocols in Section~\ref{sec:tt_impl}, 
we will focus on the first four properties. 
We will return to SMR-Progress in Section~\ref{sec:prac}
when we introduce synchrony assumptions.

\section{Tree Turtle Specification} \label{sec:tt_spec}

A tree turtle is a fault-tolerant protocol executed by the
processors in $\mathcal{P}$.
Processors construct an input and produce an output
for each tree turtle.
An input to a tree turtle is a chain $c$, which we
will denote in brackets $\langle c \rangle$.
An output of a tree turtle is a pair of chains
$\langle d, u \rangle$. 
Since tree turtles are subprotocols, 
we do not assume here that all processors---including non-crashed ones---will
 execute a given tree turtle.
We show in Section~\ref*{sec:tt_to_smr} that in
our construction, all correct processors will in fact
execute each tree turtle.

A tree turtle must satisfy the following properties:
\begin{definition}[Tree Turtle Specification] \label{def:tt_spec}
\par\noindent $ $
\begin{itemize}[leftmargin=1.25em]
\item \emph{Turtle-Termination}: 
if each correct processor constructs an input, 
then eventually each correct processor produces an output.
\item \emph{Turtle-Agreement}: 
for any two outputs $\langle d, u \rangle$ and 
$\langle d', u' \rangle$,
$d \preceq u'$ and $d' \preceq u$.
\item \emph{Turtle-Unanimity}:
for any chain $w$, 
if $w \preceq c$ for all inputs $\langle c \rangle$,
then $w \preceq d$ for all outputs $\langle d, u \rangle$.  
\item \emph{Turtle-Validity}:
if some processor produces an 
output $\langle d, u \rangle$, then some processor must
have produced an input $\langle c \rangle$ such that $u \preceq c$.
\end{itemize}
\end{definition}

Turtle-Agreement ensures agreement between any two outputs of a turtle:
note that if both $d \preceq u$ and $d' \preceq u$, then either 
$d \preceq d'$ or $d' \preceq d$. 
Further, the case where $\langle d, u \rangle = \langle d', u' \rangle$
implies that $d \preceq u$ for each output.
Turtle-Unanimity ensures that if there is an agreement 
in the inputs (i.e., proposals) to a turtle, 
then that agreement is reflected in the outputs. 
Turtle-Validity means that each output of a turtle be a prefix of 
an input to that turtle.
Unlike SMR protocols, we are able to guarantee the liveness properties
for tree turtle protocols (Turtle-Termination).

\section{Tree Turtles All The Way Down} \label{sec:tt_to_smr}

Tree turtle protocols can be composed with each other
to implement a protocol
satisfying SMR-Agreement, SMR-Validity, SMR-Relay, and SMR-Monotonicity.

The construction works as follows.
All processors in $\mathcal{P}$ are configured with instructions to execute
the same unbounded sequence of tree turtles numbered $1, 2, 3, \dots$, and so on.
For convenience, we will extend the tree turtle inputs and outputs $\langle i, c \rangle$ and $\langle i, d, u \rangle$ to now include the tree turtle number~$i$.
We assume that each processor is initialized with
the tuple $\langle 0, \bot, \bot \rangle$
as the output for the non-existent tree turtle 0.
Then the processors begin executing the tree turtles in succession.
If a processor produces tree turtle output $\langle i, d, u \rangle$, 
that processor decides the chain $d$. 
It then proposes a new chain $c$,
selecting $c$ such that $u \preceq c$,
and constructs an input 
$\langle i+1, c \rangle$ for the next tree turtle $i+1$. 

By choice of $u \preceq c$, this construction ensures the following:

\begin{lemma} \label{lemma:smr_safety_lemma}
If a processor produces the output $\langle i, d, u \rangle$ for tree turtle $i$,
then for all inputs to subsequent tree turtles $\langle j, c \rangle$
where $i < j$, it must be that $d \preceq c$.
\end{lemma}
\begin{proof}
Suppose $\langle j, c \rangle$ is an input to tree turtle $j$
made by a processor $p$.
Proceed by induction on $j - i$.
If $j = i + 1$, then by the above construction,
$p$ must have produced $\langle i, d, u \rangle$ 
as the output of tree turtle $i$ for some
chain $u \preceq c$. By Turtle-Agreement, we also know that $d \preceq u$
which implies $d \preceq c$. 
In the inductive case, we again know
that $p$ must have produced an output $\langle j-1, d', u' \rangle$ 
with $d' \preceq u' \preceq c$ in the previous tree turtle.
The inductive hypothesis gives us that $d \preceq c'$ for all
inputs to tree turtle $j-1$ $\langle j-1, c' \rangle$. 
Thus, the condition for Turtle-Unanimity is satisfied
for tree turtle $j-1$, and so it must be that $d \preceq d' \preceq c$.
\end{proof}

Further, we can use an inductive argument 
to see that each correct processor will eventually complete each tree turtle.

\begin{lemma}\label{lemma:smr_liveness_lemma}
Every correct processor eventually produces an output for each tree turtle.
\end{lemma}
\begin{proof}
Trivially, each non-crashed processor constructs an input to tree turtle 1. 
Thus, Turtle-Termination ensures that the base case of the induction is satisfied.
In the inductive case, each non-crashed processor
will use its output of tree turtle $i$ to construct an input to tree turtle $i+1$,
and so an analogous argument applies.
\end{proof}

Now we proceed to the main proof:

\begin{theorem}
The composition of tree turtles implements a protocol satisfying SMR-Agreement, SMR-Validity, SMR-Relay, and SMR-Monotonicity.
\end{theorem}

\paragraph{SMR-Agreement}  
Suppose that two processors $p$ and $p'$ decide chains
based on their outputs $\langle i, d, u \rangle$ and $\langle j, d', u' \rangle$, respectively.
First, suppose that $i = j$.
By Turtle-Agreement, $d \preceq u'$, and we also know that $d' \preceq u'$.
Then, since $d$ and $d'$ are prefixes of the same chain $u'$,
it must be that either $d \preceq d'$ or $d' \preceq d$. Thus, the decided values agree.
Now suppose that $i < j$.
By Lemma~\ref{lemma:smr_safety_lemma}, we know that $d \preceq c$ for all
inputs $\langle j, c \rangle$ to tree turtle $j$. 
By Turtle-Unanimity, we have that $d \preceq d'$. 

\paragraph{SMR-Validity:}  
Suppose that some processor $p$
decides $d$. This means that $p$ produces an output 
$\langle i, d, u \rangle$ for some tree turtle $i$ and chain $u$. 
By Turtle-Validity, there must exist an input 
$\langle i, c \rangle$ to the same tree turtle such that $u \preceq c$. 
So, $c$ was proposed by some processors.
And since we must have $d \preceq u$, we know $d \preceq c$. 

\paragraph{SMR-Relay:}
If a correct processor decides $d$ as a result of its 
output of tree turtle $i$, then by Lemma~\ref{lemma:smr_safety_lemma} and
Turtle-Unanimity, any processor that completes tree turtle $i+1$
will decide $d$ or an extension of $d$,
and Lemma~\ref{lemma:smr_liveness_lemma}
gives us that all correct processors will
do exactly as such. 

\paragraph{SMR-Monotonicity:}
If a processor decides $d$ as a result of its output of tree turtle $i$,
then by Lemma~\ref{lemma:smr_safety_lemma} and Turtle-Unanimity,
any processor that completes tree turtle $i+1$ will 
decide $d$ or an extension of $d$.
So by induction, any later decision will be monotonically extending $d$.

\hfill\qedsymbol

The above protocol does
not guarantee SMR-Progress.
However, the fact that the correct processors 
eventually complete each turtle
(Lemma~\ref{lemma:smr_liveness_lemma}) can be used to make
an auxiliary argument for SMR-Progress for a specific protocol
(for instance, using a probabilistic termination argument).

\section{Tree Turtle Implementations} \label{sec:tt_impl}

Here we present two possible tree turtle implementations.

\subsection{One-Step Tree Turtle} \label{sec:onestep_tt}

Our first tree turtle protocol uses only a single round of communication
between processors, making it a \emph{one-step}
protocol~\cite{BGMR01}. 
It requires a quorum system satisfying $3$-intersection.

A processor $p$ executing the One-Step Tree Turtle protocol proceeds as follows:
\begin{enumerate}
\item[1a.] $p$ produces an input $\langle c \rangle$ and
broadcasts it to all processors, including itself;
\item[1b.] $p$ waits to receive inputs $\langle c_s \rangle$
from all processors $s$ in any quorum $Q_p$;
\item[1c.] $p$ produces $\langle d, u \rangle$, where:
\begin{itemize}
\item[i.] $d = \bigsqcap_{s \in Q_p} c_s$;
\item[ii.] Let $C_p = \big\{ \bigsqcap_{s \in Q_p \cap Q} c_s  ~|~ Q \in \mathcal{Q} \big\}$. Then $u = \max(C_p)$.
\end{itemize}
\end{enumerate}

That is, $d$ is simply the longest common prefix of
the received proposals.  To compute $u$, $p$ considers all quorums $Q$.
For each such quorum, $p$ determines the longest
common prefix on the proposals 
it received from the processors in $Q$.
We show below that these subchains agree with one
another. 

\begin{lemma} \label{lemma:onestep_totally_ordered}
All elements of $C_p$ agree. 
\end{lemma}
\begin{proof}
All of the elements in $C_p$ are computed by taking the
intersection of $Q_p$ with another quorum.
Let $Q, Q'$ be any two quorums and $x, x'$ elements of $C_p$ 
where $x = \bigsqcap_{s \in Q_p \cap Q} c_s$,
and $x' = \bigsqcap_{s \in Q_p \cap Q'} c_s$.
Since $\mathcal{Q}$ satisfies $3$-intersection,
$Q_p \cap Q \cap Q'$ is non-empty.
Let $r$ be some processor in this intersection.
By the use of $\sqcap$, we have that $x \preceq c_r$ and $x' \preceq c_r$. 
This means that either $x \preceq x'$ or $x' \preceq x$. 
Thus, all elements of $C_p$ agree with each other. 
\end{proof}
Since $\preceq$ is transitive, $C_p$ has a maximum element
according to $\preceq$, and so the computation of $u$ is well-defined.
Further, all elements of $C_p$ are extensions of $d$:

\begin{lemma} \label{lemma:onestep_d_prefix_C}
For all $x \in C_p$, $d \preceq x$.  
\end{lemma}
\begin{proof}
For any quorum $Q$, we can observe that $Q_p \cap Q \subseteq Q_p$.
The longest common prefix over a subset of inputs belonging to processors
$Q_p$ is at least as long as the longest common prefix over $Q_p$.
Thus, $d \preceq x$ for any $x \in C_p$. 
\end{proof}

We now show that the protocol implements a tree turtle.

\begin{theorem}
The One-Step Tree Turtle protocol satisfies the tree turtle specification (Definition~\ref{def:tt_spec}).
\end{theorem}

\paragraph{Turtle-Termination:} 
Suppose that all correct processors construct an input to the tree turtle.
The only point in the protocol where a given correct processor $p$ will wait
is to receive messages from a quorum of processors at step 1b. 
Since there is assumed to be a quorum $Q^*$ that consists entirely of correct processors, 
and the network reliably delivers messages between correct processors, 
$p$ will need to wait no longer than it takes for the messages from all
processors in $Q^*$ to be delivered.
Thus, $p$ will be able to complete the turtle at step 1c.

\paragraph{Turtle-Agreement:} Suppose that two processors $p$ and $p'$ 
produce $\langle d, u \rangle$ and $\langle d', u' \rangle$,
respectively. For all processors $r \in Q_p \cap Q_{p'}$, 
both $p$ and $p'$ received the proposal $c_r$ from $r$.
Letting $x = \bigsqcap_{r \in Q_p \cap Q_{p'}} c_r$,
we see that $x$ is present in both $C_p$ and $C_{p'}$.
By Lemma~\ref{lemma:onestep_d_prefix_C},
$d \preceq x$, and by the maximality of $u'$ over $C_{p'}$, $x \preceq u'$.
So, $d \preceq u'$ by transitivity.
The same argument can be used to show that $d' \preceq u$.

\paragraph{Turtle-Unanimity:} Suppose that there exists 
a common prefix $w \preceq c$ for 
all inputs $\langle c \rangle$ to the turtle. 
Because of step 1b in the protocol,
all of the $c_s$ values used to compute $\langle d, u \rangle$
came from inputs $\langle c_s \rangle$.
This means that $w$ is a prefix of each $c_s$, and so $w$ must be a
(not necessarily strict) prefix of the longest common prefix
$d = \bigsqcap_{s \in Q_p} c_s$ of the proposals.

\paragraph{Turtle-Validity:} Suppose that 
$p$ produces an output $\langle d, u \rangle$.
We know that $u = \bigsqcap_{s \in Q_p \cap Q} c_s$
is a prefix of all proposals made by processors 
in $Q_p \cap Q$ for some quorum $Q$.
So, taking any processor $r$ in the 
intersection $Q_p \cap Q$, 
we know that $u \preceq c_r$, where $\langle c_r \rangle$
was the input produced by $r$. \hfill\qedsymbol

\subsection{Lower-Bound Tree Turtle} \label{sec:lowerbound_tt}

Now we will present a second tree turtle protocol, the
Lower-Bound Tree Turtle. This protocol meets the lower bound
on the intersection properties of the quorum system needed
to solve SMR:
namely, 2-intersection in the crash failure case~\cite{BT83}.
With these weaker assumptions about the quorum system, we can design
a protocol that makes it safe for processors to output a chain based on 
messages from all processors in a quorum under the condition that the chains
in all such messages agree. Satisfying this condition requires an additional round of communication.

A processor $p$ executing the Lower-Bound Tree Turtle protocol proceeds as follows:
\begin{enumerate}
\item[1a.] $p$ produces an input $\langle c \rangle$ and
broadcasts it to all processors, including itself;
\item[1b.] $p$ waits to receive inputs $\langle c_s \rangle$
from all processors $s$ in any quorum $Q^1_p$;
\item[1c.] $p$ computes $x = \bigsqcap_{s \in Q^1_p} c_s$;
\item[2a.] $p$ broadcasts $\langle x \rangle$ to all processors, including itself;
\item[2b.] $p$ waits to receive messages $\langle x_s \rangle$
from all processors $s$ in a quorum $Q^2_p$; 
\item[2c.] $p$ produces $\langle d, u \rangle$, where
$d = \min\limits_{s \in Q^2_p} x_s$ and
$u = \max\limits_{s \in Q^2_p} x_s$.
\end{enumerate}

First we show that all chains computed in
step 1c agree:

\begin{lemma} \label{lemma:lowerbound_agreement}
If processors $p$ and $p'$ compute
$x$ and $x'$ in step 1c respectively,
then $x \simeq x'$.
\end{lemma}
\begin{proof}
Since there must be some processor $r$ in the 
intersection $Q^1_p \cap Q^1_{p'}$,
both $p$ and $p'$ received the same chain $c_r$ from $r$.
Then, since both $x$ and $x'$ are prefixes of $c_r$,
$x$ and $x'$ must agree.
\end{proof}

Note that the preceding lemma shows that the rules to compute $d$ and $u$
are well-defined since all of the chains $x_s$ for $s \in Q_p^2$ must agree. 
The next lemma is trivial and 
shows that a common prefix to all proposals 
is preserved in the chains computed at step 1c.

\begin{lemma} \label{lemma:lowerbound_unanimity}
If there exists a chain $w$ such that $w \preceq c$ for
all inputs $\langle c \rangle$, then for any
processor $p$ that computes $x$ in step 1c, $w \preceq x$.
\end{lemma}

We now show that the Lower-Bound protocol implements the specification of a tree turtle. 

\begin{theorem}
The Lower-Bound Tree Turtle protocol satisfies the tree turtle specification (Definition~\ref{def:tt_spec}).
\end{theorem}

\paragraph{Turtle-Agreement:} Suppose that two processors $p$ and $p'$ 
produce $\langle d, u \rangle$ and $\langle d', u' \rangle$, respectively.
There must be some processor $r$ in the intersection $Q_p^2 \cap Q_{p'}^2$,
and both $p$ and $p'$ received the same chain $x_r$ from $r$.
By Lemma~\ref{lemma:lowerbound_agreement} 
and the minimality of $d$, it must be that $d \preceq x_r$.
And by the maximality of $u'$, we have that $x_r \preceq u'$.
Therefore, $d \preceq u'$. The same argument can be used to show that $d' \preceq u$.

\paragraph{Turtle-Unanimity:} Suppose that there exists a chain $w$ 
such that $w \preceq c$ for all inputs $\langle c \rangle$,
and processor $p$ produces output $\langle d, u \rangle$.
Lemma~\ref{lemma:lowerbound_unanimity} shows that $w$ will be 
a prefix of the minimum chain $d$ received by $p$ in step 2b. \hfill\qedsymbol

\vspace{1em}
\textbf{\emph{Turtle-Termination}} and \textbf{\emph{Turtle-Validity}} are
similar to the proof for One-Step Tree Turtle (Section~\ref{sec:onestep_tt}).

\subsection{Message Size}

The protocol we discussed uses messages 
containing chains that represent the entire
state machine history. As the size of this history grows, 
this quickly becomes impractical.
However, it is not necessary for a processor to broadcast
a chain in subsequent turtles 
once it has decided that chain 
following the construction in Section~\ref{sec:tt_to_smr}. 

If processor $p$ decides chain $d$ as a result of its
output of a tree turtle, then any other processor $p'$ 
that outputs $\langle d', u' \rangle$ from the same tree turtle
will have $d \preceq u'$ by Turtle-Agreement. Thus,
$p'$ already knows the contents of the chain $d$, 
and $p$ may omit that chain prefix in its proposals to subsequent turtles.

\subsection{Heterogeneous Protocols}

The simplest way to construct an SMR protocol using the
proposed abstractions is to use a single tree turtle
protocol, instantiated an unbounded number of times.
There are other options, however.

Different tree turtle protocols may have different
normal case or worst case performance properties.
The Lower-Bound Tree Turtle,
when combined with a leader as discussed in
Section~\ref{sec:prac}, has good normal case
performance properties, but it relies on synchrony
assumptions for liveness.  A similar protocol, borrowing
ideas from the Ben-Or protocol using randomness~\cite{Ben83},
can provide termination almost surely but has bad normal
case performance.  By alternating between the two
protocols, we can achieve the best of both worlds.

\section{Leaders as an Abstraction} \label{sec:prac}

Processors may make different proposals to a tree turtle,
preventing them from being able to decide new chains.
This issue of contention has been addressed previously in
SMR protocols
by using a \emph{leader} which drives all processors 
to use the same proposals.
In existing leader-based consensus protocols,
the leader lies in the critical path of the protocol:
without a functioning leader, no decisions may be made.
We show that using tree turtles, 
leaders can be easily \emph{factored out} so
that their only role is to help the protocol towards
making decisions.

Leaders can be introduced to
an existing tree turtle protocol as follows:
\begin{itemize}
\item The leader $\ell$ for tree turtle $i$
is processor with identifier $i \pmod{|\mathcal{P}|}$,
meaning the role of leader rotates through all of the processors. 
\item The leader $\ell$ for tree turtle $i$
broadcasts its input $c_{\ell}$ to tree turtle $i$.
\item All processors set a timer and wait for the leader's message. 
If a processor $p$ receives $c_{\ell}$ before the timer expires,
then it uses $c_{\ell}$ as its input to tree turtle $i$. 
Otherwise, it proceeds normally.
\item The processors double the length of the timer in each tree turtle.
\end{itemize}

Under synchronous conditions, the leader is able to eliminate contention.
The proof requires reasoning about quorums and 
messages, but these concepts generalize beyond
the protocols presented in Section~\ref{sec:tt_impl}.

\begin{lemma}
Using the above construction and the protocols in Section~\ref{sec:tt_impl},
if there exists an upper bound $\Delta$
such that a message sent between two non-crashed processors 
is delivered and processed within $\Delta$, then
there will be an unbounded number of tree turtles
where the leader's message is received by all non-crashed processors
before their timers expire.
\label{lem:t-bound}
\end{lemma}

\begin{proof}
Consider any tree turtle $i$ after the point in the execution
where the timers have become larger than $\Delta \cdot 2^{|\mathcal{P}|}$.
In general, the leader for a tree turtle may not have started executing that
tree turtle at the point when other processors begin waiting for its message. 
Applying Lemma~\ref{lemma:smr_liveness_lemma},
consider a processor $p$ that has begun tree turtle $i+1$. 
In all of our protocols, $p$ must wait to receive messages
from a quorum before completing the protocol.
Thus, must be a quorum $Q$ of processors who have begun tree turtle $i$.
Since $Q \cap Q^*$ is non-empty, there must be a tree turtle in 
the next $|\mathcal{P}|$ instances whose leader is a correct processor in $Q$.
Let $j$ be the first such instance and let $\ell$ be the leader for tree turtle $j$.

Since $\ell$ has already begun tree turtle $i$, it must catch up at most 
$|\mathcal{P}|$ tree turtles to reach $j$.
Let $t$ be the timer length for tree turtle $j$.
Since all timer lengths are greater than $2\Delta$,
$\ell$ will have received messages from a quorum for every tree turtle up 
to $j$. 
A period of $\Delta$ is sufficient for other processors to receive $\ell$'s
input to tree turtle $j$. 
However, this does not include the timers that $\ell$ must set
for tree turtles $i, \dots, j - 1$.
Since the timer lengths are doubled after each tree turtle, 
the total time required is
$\Delta + (t/2^{|\mathcal{P}|} + t/2^{|\mathcal{P}|-1} \cdots + t/2) = \Delta + (1 - 2^{-|\mathcal{P}|})t$. 
Since $t > \Delta \cdot 2^{|\mathcal{P}|}$, $\ell$'s message will
be received before the timers expire.
Since there are a constant number of turtles in which failures happen, 
there will be an unbounded number of such turtles. 
\end{proof}

\begin{theorem}
Using the above construction and the protocols in Section~\ref{sec:tt_impl},
if there exists an upper bound $\Delta$
such that a message sent between two non-crashed processors 
is delivered and processed within $\Delta$, then SMR-Progress is satisfied
for the composition of tree turtles (Section~\ref{sec:tt_to_smr}).
\end{theorem}

\begin{proof}

According to Lemma \ref{lem:t-bound}, there will be an unbounded number of tree turtles
where the leader's message is received before all timers expire.
Turtle-Unanimity provides that the processors who complete the turtle
will decide the leader's chain.
It follows that correct processors will always eventually be able to decide a new chain.

\end{proof}

\section{Byzantine Fault Tolerance} \label{sec:bft}

In this section, we extend  
tree turtles to make them Byzantine fault-tolerant (BFT).
We also modify the two protocols in
Section~\ref{sec:tt_impl} to turn them into BFT-tree turtles.

\subsection{Modified System Model}

\subsubsection*{Failures.}

The failure model will now provide that 
a faulty processor may not only crash but also behave arbitrarily:
it can deviate from the protocol, send arbitrary or
contradictory messages, and fail to send messages to some or all processors.
Such a processor is referred to as \emph{Byzantine}. 

\subsubsection*{Digital signatures.}
Because Byzantine processors can make up arbitrary data, our
protocols will use cryptographic digital signatures.
We use $\dsig{m}{p}$ to denote a digital signature of data object $m$
created by processor $p$. 
We assume that no processors can forge the signatures 
of correct processors, so that if a processor $q$ witnesses both $m$ and $\dsig{m}{p}$,
then $q$ knows that if $p$ is correct, then $p$ produced $m$.

\subsection{BFT-Tree Turtles \& BFT-SMR} 

Designing BFT-tree turtles necessitates re-defining both SMR and tree turtles for the BFT setting.
As before, processors executing an SMR protocol propose and
decide chains. The BFT-SMR properties are only slight variations 
of their crash-tolerant counterparts. Now, we only concern 
ourselves with the decisions made by correct processors.

\begin{definition}[Byzantine Fault-Tolerant State Machine Replication Specification] \label{def:bft_smr_spec}
\par\noindent $ $
\begin{itemize}[leftmargin=1.25em]
\item \emph{BFT-SMR-Agreement}:
if a correct processor $p$ decides chain $d$ and a correct processor $p'$
decides chain $d'$, then $d \simeq d'$ (even if $p = p'$).
\item \emph{BFT-SMR-Validity}:
if a correct processor decides chain $d$, 
then some processor proposed a chain $c$ such that $d \preceq c$.
\item \emph{BFT-SMR-Relay}:
if a correct processor decides chain $d$, then eventually all correct processors
decide $d$ or an extension of $d$.
\item \emph{BFT-SMR-Monotonicity}:
if a correct processor decides chain $d$ and later another chain $d'$, then $d \preceq d'$.
\item \emph{BFT-SMR-Progress}: 
if a correct processor proposes a chain that it has not decided before, then
that processor eventually decides a chain it has not decided before.
\end{itemize}
\end{definition}

To define BFT-tree turtles, however, we will take a different approach.
Rather than only specifying the behavior of correct processors, 
we will instead strengthen tree turtle inputs and outputs 
by exclude some arbitrary behavior. We will call use the terms
\emph{BFT-inputs} and \emph{BFT-outputs}, respectively, 
to distinguish them from the earlier definitions.

To allow protocols to strengthen BFT-inputs and BFT-outputs,
they will include an additional, abstract piece of data called \emph{evidence}.
The implementation of evidence is determined a specific protocol
as needed to satisfy the specification for a BFT-tree turtle.
While both of our protocols will use digital signatures in evidence,
other common techniques for turning crash-tolerant SMR protocols into BFT protocols, 
such as requiring more quorum intersection, private network links, or using 
additional rounds of communication, may be used.
In fact, a BFT-tree turtle protocol will only define the evidence 
needed for a BFT-output of the protocol 
since BFT-inputs are constructed externally from the protocol---it 
will turn out from our composition of BFT-tree turtles into BFT-SMR
that the evidence in a BFT-input will be defined by
the previous tree turtle in the execution.

In a BFT-tree turtle, a processor constructs a BFT-input $\langle c, e_c \rangle$
consisting of a chain $c$ and evidence $e_c$, and produces a BFT-output
$\langle d, u, e_{d,u} \rangle$ with a pair of chains $d,u$ and evidence $e_{d,u}$.

\begin{definition}[Byzantine Fault-Tolerant Tree Turtle Specification] 
    \label{def:bft_tt_spec}
\par\noindent $ $
\begin{itemize}[leftmargin=1.25em]
\item \emph{BFT-Turtle-Termination}: 
if each correct processor constructs a BFT-input, 
then eventually each correct processor produces a BFT-output.
\item \emph{BFT-Turtle-Agreement}: 
for any two BFT-outputs $\langle d, u, e_{d,u} \rangle$ and 
$\langle d', u', e_{d',u'} \rangle$,
$d \preceq u'$ and $d' \preceq u$.
\item \emph{BFT-Turtle-Unanimity}:
for any chain $w$, 
if $w \preceq c$ for all BFT-inputs $\langle c, e_{c} \rangle$,
then $w \preceq u$ for all BFT-outputs $\langle d, u, e_{d,u} \rangle$.  
\item \emph{BFT-Turtle-Validity}:
if some processor produces a BFT-output $\langle d, u, e_{d,u} \rangle$, 
then some processor must
have produced a BFT-input $\langle c, e_{c} \rangle$ such that $u \preceq c$.
\end{itemize}
\end{definition}

Note that BFT-Turtle-Unanimity is weaker than its crash-tolerant counterpart (Definition~\ref{def:tt_spec}). 
This weaker property is easier to satisfy
in the BFT case and is sufficient for implementing 
the safety properties of BFT-SMR, although it means that BFT-SMR-Monotonicity
is no longer guaranteed.

Evidence is not intended to prevent all Byzantine behavior.
Depending on the BFT-tree turtle protocol, faulty processors may be able to produce 
multiple BFT-inputs or BFT-outputs or 
use arbitrary chains to produce a BFT-output the protocol. 
However, evidence allowes us to reason about 
BFT-inputs and BFT-outputs made by correct and Byzantine processors alike
to the extent allowed by the constraints defined by the protocol.

\subsection{BFT-Tree Turtles Implement BFT-SMR}

The composition of BFT-tree turtles is very similar to that of crash tolerant
tree turtles (Section~\ref{sec:tt_to_smr}). 
As before, all of the processors in 
$\mathcal{P}$ will execute an unbounded sequence of tree turtles in succession.
If a processor completes BFT-tree turtle $i$ and
produces a BFT-output $\langle i, d, u, e_{d,u} \rangle$, 
then that processor decides the chain $d$ \emph{only if} $d$ is longer
than its longest decided chain.
This is to ensure BFT-SMR-Monotonicity is satisfied. 
The processor then selects a new chain $c$ to propose satisfying $u \preceq c$.
It then constructs a BFT-input to the next BFT-tree turtle $i+1$
as $\langle i+1, c, \langle i, d, u, e_{d,u} \rangle \rangle$,
using its BFT-output from the previous tree turtle as the evidence in its BFT-input.

\begin{lemma} \label{lemma:bft_smr_safety_lemma}
If a processor produces BFT-output $\langle i, d, u, e_{d,u} \rangle$ 
for BFT-tree turtle $i$, then for all BFT-inputs $\langle j, c, e_{c} \rangle$ to subsequent BFT-tree turtles 
$j > i$, it must be that $d \preceq c$.
\end{lemma}
\begin{proof}
The base case is analogous to that given for Lemma~\ref{lemma:smr_safety_lemma}.
In the inductive case, if a processor $p$ constructs a
BFT-input $\langle j, c, e_{c} \rangle$ to BFT-tree turtle $j$,
then $p$ must have produced a BFT-output $\langle j-1, d', u', e_{d', u'} \rangle$ 
in the previous tree turtle which satisfies $u' \preceq c$.
The inductive hypothesis gives us that $d \preceq c'$ for all
BFT-inputs $\langle j-1, c', e_{c'} \rangle$ to BFT-tree turtle $j-1$. 
So the condition for BFT-Turtle-Unanimity is satisfied
and this gives us that $d \preceq u' \preceq c$.
\end{proof}

\begin{lemma}\label{lemma:bft_smr_liveness_lemma}
Every correct processor eventually produces a BFT-output for each BFT-tree turtle.
\end{lemma}
\begin{proof}
The proof is the same as that given for Lemma~\ref{lemma:smr_liveness_lemma},
except that the reasoning applies to all \emph{correct} processors rather than
all non-crashed processors.
\end{proof}

\begin{theorem}
The composition of BFT-tree turtles implements a protocol satisfying 
BFT-SMR-Agreement, BFT-SMR-Validity, and BFT-SMR-Monotonicity.
\end{theorem}

\paragraph{BFT-SMR-Agreement:}  
Suppose that two correct processors $p$ and $p'$ decide chains
based on their BFT-outputs $\langle i, d, u, e_{d,u} \rangle$ 
and $\langle j, d', u', e_{d',u'} \rangle$, respectively.
First, suppose that $i = j$.
By BFT-Turtle-Agreement, $d \preceq u'$ and $d' \preceq u'$, 
and thus $d \preceq d'$ or $d' \preceq d$. 
Thus, the decided values agree.
Now suppose that $i < j$.
By Lemma~\ref{lemma:smr_safety_lemma}, we know that $d \preceq c$ for all
BFT-inputs $\langle j, c, e_{c} \rangle$ to BFT-tree turtle $j$. 
By BFT-Turtle-Unanimity, we have again that $d \preceq u'$, so it follows that $d$
and $d'$ agree.

\paragraph{BFT-SMR-Validity:}  
Suppose that a correct processor $p$
decides the chain $d$. This means that $p$ produces a BFT-output 
$\langle i, d, u, e_{d,u} \rangle$ for some BFT-tree turtle $i$, chain $u$, 
and evidence $e_{d,u}$. 
By BFT-Turtle-Validity, there must exist a BFT-input 
$\langle i, c, e_{c} \rangle$ to the same tree turtle such that $u \preceq c$. 
So, $c$ was proposed by some processor.
And since we must have $d \preceq u$, we know $d \preceq c$. 

\paragraph{BFT-SMR-Monotonicity:}
If a correct processor decides $d$ as a result of its BFT-output of BFT-tree turtle $i$,
then by Lemma~\ref{lemma:smr_safety_lemma} and BFT-Turtle-Unanimity,
an inductive argument shows that any BFT-output that the processor produces
for a later BFT-tree turtle will have a chain $d'$ such that $d \preceq d'$
or $d' \preceq d$. Since correct processors will only decide a new chain
if it is an extension of its longest decided chain, it follows
that the processor's decisions are extended monotonically.

\hfill\qedsymbol

\subsection{BFT One-Step Tree Turtle}

The One-Step Tree Turtle (Section~\ref{sec:onestep_tt}) can be made into
a BFT-tree turtle with the following modifications:
\begin{itemize}
\item The quorum system $\mathcal{Q}$ must satisfy 5-intersection instead
of 3-intersection (meeting the lower bound for quorum
intersection of BFT one-step protocols~\cite{SvR08});
\item In step 1a., processor $p$ broadcasts  
$\langle \langle c, e_{c} \rangle, \dsig{c}{p} \rangle$,
i.e., its BFT-input and a digital signature for $c$;
\item In step 1b., processor $p$ must ensure that the received messages 
contain valid digital signatures and BFT-inputs $\langle c_s, e_{c_s} \rangle$,
and for any that do not, $p$ discards the message;
\item In step 1c., processor $p$ constructs the evidence $e_{d,u}$ 
as set of tuples $\{ \langle c_s, \dsig{c_s}{s} \rangle ~|~ s \in Q_p\}$
and includes it in its BFT-output $\langle d, u, e_{d,u} \rangle$;
\item A BFT-output $\langle d, u, e_{d,u} \rangle$ must contain evidence
$e_{d,u} = \{ \langle c_s, \dsig{c_s}{s} \rangle ~|~ s \in Q\}$
for some quorum $Q$ such that the chains $c_s$ allow $d$ and $u$ to be
recomputed according to the protocol at step 1c.
\end{itemize}

The proof of correctness for the BFT One-Step Tree Turtle follows
analogous reasoning to that given for the crash tolerant version.

\begin{lemma} \label{lemma:bft_onestep_totally_ordered}
All elements of $C_p$ agree. 
\end{lemma}
\begin{proof}
All of the elements in $C_p$ are computed by taking the
intersection of $Q_p$ with two other quorums.
Suppose $x, x'$ are elements of $C_p$, where 
$x = \bigsqcap_{s \in Q_p \cap Q_1 \cap Q_2} c_s$,
and $x' = \bigsqcap_{s \in Q_p \cap Q'_1 \cap Q'_2} c_s$,
for some quorums $Q_1, Q_2, Q'_1, Q'_2$.
Since the quorum system satisfies 5-intersection,
$Q_p \cap Q_1 \cap Q_2 \cap Q'_1 \cap Q'_2$
is non-empty.
Let $r$ be some processor in this intersection.
By the use of $\sqcap$, we have that $x \preceq c_r$ and $x' \preceq c_r$,
where $c_r$ is the chain from $r$ used to complete the protocol. 
This means that either $x \preceq x'$ or $x' \preceq x$. 
\end{proof}

\begin{lemma} \label{lemma:bft_onestep_d_prefix_C}
For all $x \in C_p$, $d \preceq x$.  
\end{lemma}
\begin{proof}
Observe that $Q_p \cap Q_1 \cap Q_2 \subseteq Q_p$.
The longest common prefix over a subset of $Q_p$ is at least as long as
the longest common prefix over $Q_p$.
Thus, $d \preceq x$ for any $x \in C_p$. 
\end{proof}

\begin{theorem} \label{thm:bft_onestep}
The BFT One-Step Tree Turtle protocol satisfies the
BFT-tree turtle specification (Definition~\ref{def:bft_tt_spec}).
\end{theorem}

\paragraph{BFT-Turtle-Termination:} 
Suppose that all correct processors 
construct a BFT-input for the protocol. 
Consider any correct processor $p$. By construction of the protocol, 
$p$ will wait to receive messages from a quorum of processors at step 1b. 
This is the only point in the protocol where $p$ will wait.
Since there is at least one quorum $Q^*$ that consists entirely of correct processors, 
and the network reliably delivers messages between correct processors, 
$p$ will eventually receive messages from all processors in a quorum.
Furthermore, all messages sent by processors in $Q^*$ will contain
BFT-inputs and valid digital signatures, 
as correct processors do not send bad messages in step 1a.
Thus, $p$ will not discard any of these messages
and will be able to complete the turtle at step 1c.

\paragraph{BFT-Turtle-Agreement:} Suppose that two processors $p$ and $p'$ 
produce BFT-outputs $\langle d, u, e_{d,u} \rangle$ and 
$\langle d', u', e_{d',u'} \rangle$, respectively.
Let $Q$ and $Q'$ be the quorums used to construct $e_{d,u}$
and $e_{d',u'}$, respectively.
For all (correct) processors $r \in Q \cap Q' \cap Q^*$, 
both $e$ and $e'$ must contain the same digitally signed chain
which was sent by $r$:
$\langle c_r, \dsig{c_r}{r} \rangle$ .
This is because correct processors do not send different
messages to different processors at step 1a., so $r$ 
must have sent the same message to both $p$ and $p'$. 
Furthermore, the digital signatures of correct processors cannot be forged,
so $p$ and $p'$ must have used the signature that $r$ actually constructed in step 1a.
Thus, if the chains in $e$ are used in step 1c. to compute 
$x = \bigsqcap_{s \in Q \cap Q' \cap Q^*} c_s$
and the chains in $e'$ are used to compute 
$x' = \bigsqcap_{s \in Q \cap Q' \cap Q^*} c_s$,
we would have that $x = x'$. Thus $p$ would have
$x \in C_p$ and $p'$ would have $x \in C_{p'}$.
By Lemma~\ref{lemma:onestep_d_prefix_C}, if $e$ is used to compute
$d$, we would have that $d \preceq x$.
And, if $e'$ is used to compute $u'$, it must be that $x \preceq u'$. 
Therefore $d \preceq u'$, and an analogous argument shows that $d \preceq u'$.

\paragraph{BFT-Turtle-Unanimity:} Assume that $w \preceq c$ for 
all BFT-inputs $\langle c, e_{c} \rangle$ to the protocol. 
For any BFT-output $\langle d, u, e_{d,u} \rangle$,
let $Q$ be the quorum used to construct $e_{d,u}$.
Then for all $r \in Q \cap Q^*$, the digitally signed chain
$\langle c_r, \dsig{c_r}{r} \rangle$ in $e$ must have been 
produced by the correct processor $r$. Because of step 1a., 
correct processors only sign chains $c_r$ that are part of BFT-inputs
$\langle c_r, e_{c_r} \rangle$, meaning that $w \preceq c_r$.
Thus, taking any third quorum $Q'$, if $e$ is used to compute 
$x = \bigsqcap_{s \in {Q \cap Q' \cap Q^* }} c_s$, 
it follows that $w \preceq x$.
Therefore, by Lemma~\ref{lemma:bft_onestep_totally_ordered}, 
it must be that $x \preceq u$ which implies $w \preceq u$.

\paragraph{BFT-Turtle-Validity:} 
If a processor produces a BFT-output $\langle d, u, e_{d,u} \rangle$,
then by quorum intersection it must be that
$u \preceq c_r$ for some correct processor $r \in Q^*$,
where $\langle c_r, e_{c_r} \rangle$ is the BFT-input constructed 
by $r$. 

\hfill\qedsymbol

\subsection{BFT Lower-Bound Tree Turtle}

The Lower-Bound Tree Turtle (Section~\ref{sec:lowerbound_tt}) can be made into
a BFT-tree turtle with the following modifications:
\begin{itemize}
\item The quorum system $\mathcal{Q}$ must satisfy 3-intersection;
\item In step 1a., processor $p$ broadcasts  
$\langle \langle c, e_{c} \rangle, \dsig{c}{p} \rangle$,
i.e., its BFT-input and a digital signature for $c$;
\item In step 1b., processor $p$ must ensure that the received messages 
contain valid digital signatures and BFT-inputs $\langle c_s, e_{c_s} \rangle$,
and for any that do not, $p$ discards the message;
\item In step 2a., processor $p$ broadcasts
\newline
$\langle x, \dsig{x}{p}, \{ \langle c_s, \dsig{c_s}{s} \rangle ~|~ s \in Q^1_p\} \rangle$,
i.e., the computed chain $x$, a digital signature for $x$, and all of the chains
and digital sigatures that it used to complete step 1c. of the protocol;
\item In step 2b., processor $p$ must ensure that the received messages 
contain valid digitally signed chains such that $x_s$ can be 
recomputed according to the protocol at step 1c., 
and for any that do not, $p$ discards the message;
\item In step 2c., processor $p$ constructs the evidence $e_{d,u}$ 
as set of tuples $\{ \langle x_s, \dsig{x_s}{s} \rangle ~|~ s \in Q^2_p\}$
and includes it in its BFT-output $\langle d, u, e_{d,u} \rangle$;
\item A BFT-output $\langle d, u, e_{d,u} \rangle$ must contain evidence
$e_{d,u} = \{ \langle x_s, \dsig{c_s}{s} \rangle ~|~ s \in Q\}$
for some quorum $Q$ such that the chains $c_s$ allow $d$ and $u$ to be
recomputed according to the protocol at step 2c. \emph{and}
all of the chains $x_s$ agree.
\end{itemize}

The BFT Lower-Bound Tree Turtle requires an extra condition on the
digitally signed chains included in evidence for BFT-outputs
which is not necessary for the BFT One-Step Tree Turtle. 
This ensures that the minimum and maximum of those chains is
always well-defined. Further, this property
is always satisfied by the evidence constructed by a correct processor executing
the protocol due to the following:

\begin{lemma} \label{lemma:bft_lowerbound_agreement}
If processor $p$ broadcasts a message
\newline
$m = \langle x, \dsig{x}{p}, \{ \langle c_s, \dsig{c_s}{s} \rangle ~|~ s \in Q\} \rangle$
and $p'$ broadcasts a message
$m' =\langle x', \dsig{x'}{p'}, \{ \langle c_s, \dsig{c_s}{s} \rangle ~|~ s \in Q'\} \rangle$
at step 2a. such that the digitally signed chains included in 
$m$ and $m'$ allow $x$ and $x'$ to be recomputed, respectively, according to the protocol at step 1c.,
then $x$ and $x'$ agree.
\end{lemma}
\begin{proof}
For any processor $r$ in the intersection $Q \cap Q' \cap Q^*$,
both $m$ and $m'$ will contain the same digitally signed chain
$\langle c_r, \dsig{c_r}{r} \rangle$ from $r$.
This is because correct processors only send a single message to all
processors at step 1b. of the protocol.
Thus both $x$ and $x'$ must be prefixes of $c_r$ and therefore agree with each other.
\end{proof}

\begin{lemma} \label{lemma:bft_lowerbound_unanimity}
For any chain $w$ such that $w \preceq c$ for
all BFT-inputs $\langle c, e_{c}\rangle$,
then if processor $p$ that broadcasts
a message $m = \langle x, \dsig{x}{p}, \{ \langle c_s, \dsig{c_s}{s} \rangle ~|~ s \in Q\} \rangle$
at step 2a., such that the digitally signed chains included in 
$m$ allow $x$ to be recomputed according to the protocol at step 1c.,
then $w$ and $x$ agree. Further, if $p$ is correct, then $w \preceq x$.
\end{lemma}
\begin{proof}
Since $Q$ must contain at least one correct processor $r$ due to quorum intersection,
$m$ contains a digitally signed chain $\langle c_r, \dsig{c_r}{r} \rangle$ 
such that $w \preceq c_r$.
This is because correct processors will only produce such a signature 
in step 1a. for chains $c_r$ that are part of BFT-inputs $\langle c_r, e_{c_r} \rangle$.
Therefore $x \preceq c_r$ and so $x$ and $w$ agree.

If $p$ is correct, then $p$ will only use messages containing
BFT-inputs to compute $x$ at step 1c. of the protocol.
By assumption, $w$ is a prefix of the chains contained
in all such inputs.
\end{proof}

\begin{theorem} \label{thm:bft_lowerbound}
The BFT Lower-Bound Tree Turtle protocol satisfies the
BFT-tree turtle specification (Definition~\ref{def:bft_tt_spec}).
\end{theorem}

\paragraph{BFT-Turtle-Agreement:} 
Suppose that two processors $p$ and $p'$ 
produce BFT-outputs $\langle d, u, e_{d,u} \rangle$
and $\langle d', u', e_{d', u'} \rangle$, respectively.
Let $Q$ and $Q'$ 
be the quorums over which the 
evidence $e_{d,u}$ and $e_{d',u'}$ are constructed,
respectively. Then, there must be a correct
processor $r$ in the intersection $Q \cap Q'$,
and both $e$ and $e'$ must include the chain 
$x_r$ broadcast by $r$ in step 2b.
By the minimality of $d$, it must be that $d \preceq x_r$.
And by the maximality of $u'$, we have that $x_r \preceq u'$.
Therefore, $d \preceq u'$, and an analogous argument shows $d' \preceq u$.

\paragraph{BFT-Turtle-Unanimity:} 
Suppose that $w \preceq c$ for all BFT-inputs $\langle c, e_c \rangle$,
and that processor $p$ produces BFT-output $\langle d, u, e_{d,u} \rangle$.
Letting $Q$ be the quorum over which $e_{d,u}$ is constructed,
quorum intersection ensures that there is at least one correct processor $r$ in $Q$.
By Lemma \ref{lemma:lowerbound_unanimity}, if $r$ broadcasts 
a message including the chain $x_r$ in step 2b, then $w \preceq x_r$. 
Further, $x_r$ must be equal to, or a prefix of, the maximum chain in $e_{d,u}$. 
Thus, $w \preceq x_r \preceq u$.

\paragraph{Turtle-Termination \& Turtle-Validity:} the proofs are
analogous to those given for Theorem~\ref{thm:bft_onestep}.

\hfill\qedsymbol

\section{Related Work}

Abstracting the problem of achieving consensus on a single value to 
entire sequences was previously applied to Paxos in Generalized Paxos~\cite{GENPAXOS}. 
This work further generalizes chains to partial-orders of values in which 
non-interfering values commute. The HotStuff protocol~\cite{HotStuff}
conceptualizes the state machine history as a tree, but it only extends the
tree by a single node at a time. The same can be said of blockchain protocols.

The idea of heterogeneous SMR protocols is similar to protocols which 
have different ``modes''. Typically there is one mode which is considered 
the normal operation of the protocol and 
another designed for fast-tracking decisions under best-case conditions. 
One such example is Fast Paxos~\cite{FASTPAXOS} 
which is able to skip a round of communication in periods without contention.
These modes, however, are usually considered to be part of a single protocol
instead of separate protocols satisfying a common specification.

\section{Conclusion \& Future Work}

SMR protocols have been around for over
thirty years. We revisit the structure of these 
protocols and propose new abstractions---\emph{trees and turtles}---for
the design of modular SMR protocols.  

While this paper did not discuss the performance of tree turtle protocols,
we believe that they have potential to be performant through their ability
to drive long extensions to the state machine history in a few rounds of communication.
Future work could include an empirical analysis of their
performance. Further techniques
for optimization can also be investigated such as pipelining chains 
from different rounds of the protocol simultaneously.

There is work in progress by the authors 
on building formally verified consensus and SMR protocols
in both Dafny and Coq based on the presented abstractions
due to their simplicity and efficiency.
We also expect our abstractions to support different flavors of SMR,
such as ordered consensus and heterogeneous consensus,
by utilizing the additional structure of trees and
the flexibility of the tree turtle composition.
Further generalizations of trees and chains into more general
algebriac structures are also possible.

\begin{acks}
The authors would like to thank Pierre Sutra 
and the anonymous reviewers for their helpful suggestions.
\end{acks}

\bibliographystyle{ACM-Reference-Format}
\bibliography{references.bib}


\begin{thebibliography}{14}


\ifx \showCODEN    \undefined \def \showCODEN     #1{\unskip}     \fi
\ifx \showDOI      \undefined \def \showDOI       #1{#1}\fi
\ifx \showISBNx    \undefined \def \showISBNx     #1{\unskip}     \fi
\ifx \showISBNxiii \undefined \def \showISBNxiii  #1{\unskip}     \fi
\ifx \showISSN     \undefined \def \showISSN      #1{\unskip}     \fi
\ifx \showLCCN     \undefined \def \showLCCN      #1{\unskip}     \fi
\ifx \shownote     \undefined \def \shownote      #1{#1}          \fi
\ifx \showarticletitle \undefined \def \showarticletitle #1{#1}   \fi
\ifx \showURL      \undefined \def \showURL       {\relax}        \fi
\providecommand\bibfield[2]{#2}
\providecommand\bibinfo[2]{#2}
\providecommand\natexlab[1]{#1}
\providecommand\showeprint[2][]{arXiv:#2}

\bibitem[\protect\citeauthoryear{Ben-Or}{Ben-Or}{1983}]%
        {Ben83}
\bibfield{author}{\bibinfo{person}{M. Ben-Or}.}
  \bibinfo{year}{1983}\natexlab{}.
\newblock \showarticletitle{Another advantage of free choice: {C}ompletely
  Asynchronous Agreement Protocols}. In \bibinfo{booktitle}{\emph{Proc. of the
  2nd ACM Symp. on Principles of Distributed Computing}}. ACM SIGOPS-SIGACT,
  \bibinfo{publisher}{{ACM} Press}, \bibinfo{address}{Montreal, Quebec},
  \bibinfo{pages}{27--30}.
\newblock


\bibitem[\protect\citeauthoryear{Bracha and Toueg}{Bracha and Toueg}{1983}]%
        {BT83}
\bibfield{author}{\bibinfo{person}{G. Bracha} {and} \bibinfo{person}{S.
  Toueg}.} \bibinfo{year}{1983}\natexlab{}.
\newblock \showarticletitle{Resilient Consensus Protocols}. In
  \bibinfo{booktitle}{\emph{Proc. of the 2nd ACM Symp. on Principles of
  Distributed Computing}}. \bibinfo{publisher}{ACM SIGOPS-SIGACT},
  \bibinfo{address}{Montreal, Quebec}, \bibinfo{pages}{12--26}.
\newblock


\bibitem[\protect\citeauthoryear{Brasileiro, Greve, Most\'{e}faoui, and
  Raynal}{Brasileiro et~al\mbox{.}}{2001}]%
        {BGMR01}
\bibfield{author}{\bibinfo{person}{F.V. Brasileiro}, \bibinfo{person}{F.
  Greve}, \bibinfo{person}{A. Most\'{e}faoui}, {and} \bibinfo{person}{M.
  Raynal}.} \bibinfo{year}{2001}\natexlab{}.
\newblock \showarticletitle{Consensus in One Communication Step}. In
  \bibinfo{booktitle}{\emph{PaCT '01: Proceedings of the 6th International
  Conference on Parallel Computing Technologies}}.
  \bibinfo{publisher}{Springer-Verlag}, \bibinfo{address}{London, UK},
  \bibinfo{pages}{42--50}.
\newblock
\showISBNx{3-540-42522-5}


\bibitem[\protect\citeauthoryear{Castro and Liskov}{Castro and Liskov}{1999}]%
        {CL99}
\bibfield{author}{\bibinfo{person}{M. Castro} {and} \bibinfo{person}{B.
  Liskov}.} \bibinfo{year}{1999}\natexlab{}.
\newblock \showarticletitle{Practical {B}yzantine {F}ault {T}olerance}. In
  \bibinfo{booktitle}{\emph{Proc. of the 3rd Symposium on Operating Systems
  Design and Implementation (OSDI'99)}}. \bibinfo{publisher}{USENIX},
  \bibinfo{address}{New Orleans, LA}.
\newblock


\bibitem[\protect\citeauthoryear{Dwork, Lynch, and Stockmeyer}{Dwork
  et~al\mbox{.}}{1988}]%
        {DLS88}
\bibfield{author}{\bibinfo{person}{C. Dwork}, \bibinfo{person}{N. Lynch}, {and}
  \bibinfo{person}{L. Stockmeyer}.} \bibinfo{year}{1988}\natexlab{}.
\newblock \showarticletitle{Consensus in the Presence of Partial Synchrony}.
\newblock \bibinfo{journal}{\emph{J. ACM}} \bibinfo{volume}{35},
  \bibinfo{number}{2} (\bibinfo{date}{April} \bibinfo{year}{1988}),
  \bibinfo{pages}{288--323}.
\newblock


\bibitem[\protect\citeauthoryear{Fischer, Lynch, and Patterson}{Fischer
  et~al\mbox{.}}{1983}]%
        {FLP83}
\bibfield{author}{\bibinfo{person}{M.J. Fischer}, \bibinfo{person}{N.A. Lynch},
  {and} \bibinfo{person}{M.S. Patterson}.} \bibinfo{year}{1983}\natexlab{}.
\newblock \showarticletitle{Impossibility of Distributed Consensus with one
  Faulty Process}. In \bibinfo{booktitle}{\emph{Proceedings of the 2nd
  Symposium on Principles of Database Systems (PODS'83)}}.
  \bibinfo{publisher}{ACM SIGACT-SIGMOD}, \bibinfo{address}{Atlanta, GA}.
\newblock


\bibitem[\protect\citeauthoryear{Junqueira and Marzullo}{Junqueira and
  Marzullo}{2005}]%
        {JM05}
\bibfield{author}{\bibinfo{person}{F.P. Junqueira} {and} \bibinfo{person}{K.
  Marzullo}.} \bibinfo{year}{2005}\natexlab{}.
\newblock \showarticletitle{Replication predicates for dependent-failure
  algorithms}. In \bibinfo{booktitle}{\emph{Proceedings of the 11th Euro-Par
  Conference}} \emph{(\bibinfo{series}{Lecture Notes on Computer Science},
  \bibinfo{number}{3648})}. \bibinfo{publisher}{Springer-Verlag},
  \bibinfo{address}{Monte de Caparica, Portugal}, \bibinfo{pages}{617--632}.
\newblock


\bibitem[\protect\citeauthoryear{Junqueira, Reed, and Serafini}{Junqueira
  et~al\mbox{.}}{2011}]%
        {ZAB}
\bibfield{author}{\bibinfo{person}{F.P. Junqueira}, \bibinfo{person}{B.C.
  Reed}, {and} \bibinfo{person}{M. Serafini}.} \bibinfo{year}{2011}\natexlab{}.
\newblock \showarticletitle{{ZAB}: High-performance broadcast for
  primary-backup systems}. In \bibinfo{booktitle}{\emph{2011 IEEE/IFIP 41st
  International Conference on Dependable Systems and Networks (DSN)}}.
  \bibinfo{pages}{245--256}.
\newblock
\urldef\tempurl%
\url{https://doi.org/10.1109/DSN.2011.5958223}
\showDOI{\tempurl}


\bibitem[\protect\citeauthoryear{Lamport}{Lamport}{1998}]%
        {L98}
\bibfield{author}{\bibinfo{person}{L. Lamport}.}
  \bibinfo{year}{1998}\natexlab{}.
\newblock \showarticletitle{The Part-Time Parliament}.
\newblock \bibinfo{journal}{\emph{Trans. on Computer Systems}}
  \bibinfo{volume}{16}, \bibinfo{number}{2} (\bibinfo{year}{1998}),
  \bibinfo{pages}{133--169}.
\newblock


\bibitem[\protect\citeauthoryear{Lamport}{Lamport}{2005}]%
        {GENPAXOS}
\bibfield{author}{\bibinfo{person}{L. Lamport}.}
  \bibinfo{year}{2005}\natexlab{}.
\newblock \showarticletitle{Generalized consensus and Paxos}.
\newblock  (\bibinfo{year}{2005}).
\newblock


\bibitem[\protect\citeauthoryear{Lamport}{Lamport}{2006}]%
        {FASTPAXOS}
\bibfield{author}{\bibinfo{person}{L. Lamport}.}
  \bibinfo{year}{2006}\natexlab{}.
\newblock \showarticletitle{Fast Paxos}.
\newblock \bibinfo{journal}{\emph{Distrib. Comput.}} \bibinfo{volume}{19},
  \bibinfo{number}{2} (\bibinfo{date}{oct} \bibinfo{year}{2006}),
  \bibinfo{pages}{79–103}.
\newblock
\showISSN{0178-2770}
\urldef\tempurl%
\url{https://doi.org/10.1007/s00446-006-0005-x}
\showDOI{\tempurl}


\bibitem[\protect\citeauthoryear{Schneider}{Schneider}{1990}]%
        {S90}
\bibfield{author}{\bibinfo{person}{F.B. Schneider}.}
  \bibinfo{year}{1990}\natexlab{}.
\newblock \showarticletitle{Implementing fault-tolerant services using the
  state machine approach: A tutorial}.
\newblock \bibinfo{journal}{\emph{Comput. Surveys}} \bibinfo{volume}{22},
  \bibinfo{number}{4} (\bibinfo{date}{Dec.} \bibinfo{year}{1990}),
  \bibinfo{pages}{299--319}.
\newblock


\bibitem[\protect\citeauthoryear{Song and Van~Renesse}{Song and
  Van~Renesse}{2008}]%
        {SvR08}
\bibfield{author}{\bibinfo{person}{Y.J. Song} {and} \bibinfo{person}{R.
  Van~Renesse}.} \bibinfo{year}{2008}\natexlab{}.
\newblock \showarticletitle{Bosco: One-Step {B}yzantine Asynchronous
  Consensus}. In \bibinfo{booktitle}{\emph{Proc. of the 22nd International
  Symposium on DIStributed Computing (DISC'08)}}
  \emph{(\bibinfo{series}{Lecture Notes on Computer Science},
  \bibinfo{number}{5218})}. \bibinfo{publisher}{Springer-Verlag},
  \bibinfo{address}{Arcachon, France}.
\newblock


\bibitem[\protect\citeauthoryear{Yin, Malkhi, Reiter, Gueta, and Abraham}{Yin
  et~al\mbox{.}}{2019}]%
        {HotStuff}
\bibfield{author}{\bibinfo{person}{M. Yin}, \bibinfo{person}{D. Malkhi},
  \bibinfo{person}{M.K. Reiter}, \bibinfo{person}{G.G. Gueta}, {and}
  \bibinfo{person}{I. Abraham}.} \bibinfo{year}{2019}\natexlab{}.
\newblock \showarticletitle{HotStuff: BFT Consensus with Linearity and
  Responsiveness}. In \bibinfo{booktitle}{\emph{Proceedings of the 2019 ACM
  Symposium on Principles of Distributed Computing}} (Toronto ON, Canada)
  \emph{(\bibinfo{series}{PODC '19})}. \bibinfo{publisher}{Association for
  Computing Machinery}, \bibinfo{address}{New York, NY, USA},
  \bibinfo{pages}{347–356}.
\newblock
\showISBNx{9781450362177}
\urldef\tempurl%
\url{https://doi.org/10.1145/3293611.3331591}
\showDOI{\tempurl}


\end{thebibliography}

\newpage

\section*{Appendix A: Threshold Quorum Systems}

For completeness, we include how a $k$-intersecting
quorum system can be implemented (based on~\cite{JM05}).
A \emph{threshold quorum system} constructs 
quorums from all subsets of $\mathcal{P}$ above a particular size.
Let $n = |\mathcal{P}|$ be the number of processors.
Let $f$ be the maximum number of processors that may be faulty.
We can construct a $k$-intersecting quorum system
for $k \ge 1$ such that $n > k \cdot f$
as follows: the quorums are all subsets of $\mathcal{P}$
that have at least $n-f$ processors. 
Then $Q^*$ is the quorum that contains the $n-f$ correct processors. 
In a $k$-intersecting quorum system implemented with a threshold quorum system, 
increasing the value for $k$ results in an increase in 
the number of correct processors that are assumed to exist in the system.

\end{document}